\documentclass[12pt,oneside]{amsart}

\usepackage{amsmath, amssymb, amsthm, amscd, url} 

\allowdisplaybreaks[4]

\newcommand{\eval}[2][\right]{\relax
  \ifx#1\right\relax \left.\fi#2#1\rvert}

\newcommand{\pd}{{\partial}}

\newcommand{\al}{{\alpha}}
\newcommand{\la}{{\lambda}}

\newcommand{\mat}{\mathcal{M}}
\newcommand{\ub}{\underbrace}

\newcommand{\La}{{\Lambda}}

\newcommand{\er}{\eqref}
\newcommand{\cl}{\colon}

\newcommand{\zp}{\mathbb{Z}_+}

\newcommand{\ad}{{\rm ad\,}}

\newcommand{\mg}{\mathfrak{g}}

\newcommand{\vf}{\varphi}

\newcommand{\Com}{\mathbb{C}}

\newcommand{\cprime}{\/{\mathsurround=0pt$'$}}

\newtheorem{theorem}{Theorem}
\newtheorem{proposition}{Proposition}
\newtheorem{lemma}{Lemma}

\theoremstyle{definition}

\newtheorem{remark}{Remark}

\begin{document}






\title[Generalized Landau-Lifshitz systems and Lie algebras]
{Generalized Landau-Lifshitz systems  
and Lie algebras associated with higher genus curves}
\date{}

\author{S.~Igonin}
\address{Sergey Igonin \\
Department of Mathematics, University of Utrecht, P.O. Box 80010, 3508 TA Utrecht,
The Netherlands}
\email{igonin@mccme.ru}

\author{J.~van~de~Leur}
\address{Johan~van~de~Leur\\
Department of Mathematics, University of Utrecht, P.O. Box 80010, 3508 TA Utrecht,
The Netherlands}
\email{J.W.vandeLeur@uu.nl}

\author{G.~Manno}
\address{Gianni Manno \\
Universit\`a degli Studi di Milano-Bicocca, Dipartimento di Matematica e Applicazioni, Via Cozzi 53, 20125 Milano, Italy}
\email{gianni.manno@unimib.it}

\author{V.~Trushkov}
\address{Vladimir~Trushkov \\
University of Pereslavl, Sovetskaya 2, 152020 Pereslavl-Zalessky, Yaroslavl region, Russia}
\email{vladimir@trushkov.pereslavl.ru}

\begin{abstract}
The Wahlquist-Estabrook prolongation method 
allows to obtain for some PDEs a Lie algebra 
that is responsible for Lax pairs and B\"acklund transformations of certain type. 
We study the Wahlquist-Estabrook algebra 
of the $n$-dimensional generalization of the Landau-Lifshitz equation 
and construct an epimorphism from this algebra  
onto an infinite-dimensional quasigraded Lie algebra $L(n)$ of certain matrix-valued functions 
on an algebraic curve of genus $1\!+\!(n\!-\!3)2^{n-2}$. 
For $n=3,4,5$ we prove that the Wahlquist-Estabrook algebra is isomorphic 
to the direct sum of $L(n)$ and a $2$-dimensional abelian Lie algebra. 
Using these results, for any $n$ a new family of Miura type transformations (differential substitutions) 
parametrized by points of the above mentioned curve is constructed. 
As a by-product, we obtain a representation of $L(n)$ in terms of a finite number 
of generators and relations, which may be of independent interest.
\end{abstract}


\maketitle 

\section{Introduction}

In the last 25 years it has been well understood how to construct integrable PDEs  
from infinite-dimensional Lie algebras
(see, e.g.,~\cite{feher,z,gds,skr,skr-jmp} and references therein). 
The present paper addresses the inverse problem:  
given a system of PDEs, how to determine whether 
it is related to an infinite-dimensional 
Lie algebra and how to recover this Lie algebra? 
  
A partial answer to this question is provided by 
the so-called Wahlquist-Estabrook prolongation method~\cite{dodd,Prol}, 
which is an algorithmic procedure that for a given $(1+1)$-dimensional system of PDEs 
constructs a Lie algebra called the \emph{Wahlquist-Estabrook algebra}  
(or the \emph{WE algebra} in short). 
The WE algebra is responsible for Lax pairs and B\"acklund transformations of certain type. 
The method gives this algebra in terms of generators and relations. 
For some PDEs the explicit structure of the WE algebra 
was revealed, and as a result one obtained interesting infinite-dimensional 
Lie algebras (see, e.g., \cite{schief,kdv1,kn,ll} and references therein).  

In the original method of Wahlquist and Estabrook  
the obtained algebras lack any invariant coordinate-free meaning.   
Recently the WE algebra has been included in a sequence 
of Lie algebras that have a remarkable geometric interpretation: 
they are the analogue of the topological fundamental group 
for the category of PDEs~\cite{cfa,lie}. 
All finite-dimensional quotients of these Lie algebras 
have a coordinate-free meaning as symmetry algebras of certain coverings of PDEs~\cite{cfa} 
(the notion of coverings of PDEs by Krasilshchik and Vinogradov~\cite{rb,nonl} is 
a far-reaching geometric generalization of Wahlquist-Estabrook pseudopotentials). 
Moreover, an effective necessary condition for existence of a B\"acklund transformation 
connecting two given PDEs has been obtained in terms of these Lie algebras~\cite{cfa}, 
which has allowed to prove for the first time ever that some PDEs are not connected by any B\"acklund transformation~\cite{cfa}.
In our opinion, these results strongly suggest to compute and study 
the WE algebras for more PDEs. 

In this paper we apply the Wahlquist-Estabrook method to the following system 
\begin{equation}
\label{main}
S_t=\Bigl(S_{xx}+\frac32\langle S_x,S_x\rangle S\Bigl)_x+\frac32\langle S,RS\rangle S_x,
\quad \langle S,S\rangle=1,
\end{equation} 
where $S=(s^1(x,t),\dots,s^n(x,t))$ is a column-vector of dimension $n>1$, 
$\langle\cdot,\cdot\rangle$ is the standard scalar product, 
and $R=\mathrm{diag}(r_1,\dots,r_n)$ is a constant diagonal matrix 
with $r_i\neq r_j$ for $i\neq j$.   
This system was introduced in~\cite{mll} and possesses a Lax pair 
(a zero-curvature representation) parametrized by points of  
the following algebraic curve 
\begin{equation}
\label{curve}
\la_i^2-\la_j^2=r_j-r_i,\quad i,j=1,\dots,n,
\end{equation}
in the space $\Com^n$ with coordinates $\la_1,\dots,\la_n$.
According to~\cite{mll}, this curve is of genus ${1\!+\!(n\!-\!3)2^{n-2}}$. 

System~\er{main} has also an infinite number of symmetries, 
conservation laws~\cite{mll}, and a B\"acklund auto-transformation with a parameter~\cite{ll-backl}. 
Soliton-like solutions of~\er{main} can be found in~\cite{ll-backl}. 
According to~\cite{mll}, 
for $n=3$ this system coincides with the higher symmetry (the commuting flow) 
of third order for the well-known Landau-Lifshitz equation (see, e.g.,~\cite{faddeev}). 
Thus system~\er{main} is an $n$-dimensional generalization of the Landau-Lifshitz equation. 

Note that, to our knowledge, before the present paper 
the Wahlquist-Estabrook method was never applied 
to a system of PDEs connected with an algebraic curve of genus greater than $1$, 
and the explicit structure of the WE algebra was not computed 
for any system with more than two dependent variables 
(unknown functions $s^i(x,t)$). 

For arbitrary $n$ we construct an epimorphism from the WE algebra of~\er{main}
onto the infinite-dimensional Lie algebra $L(n)$ of certain $\mathfrak{so}_{n,1}$-valued 
functions on the curve~\er{curve}. 
Note that $L(n)$ is not graded, but is quasigraded~\cite{quasigraded}  
\begin{gather*}
L(n)=\bigoplus_{i=1}^\infty\bar L_i,\quad\  
[\bar L_i,\bar L_j]\subset \bar L_{i+j}+\bar L_{i+j-2},\\
\quad\dim\bar L_{2k-1}=n,\quad\dim\bar L_{2k}=\frac{n(n-1)}{2}\quad\forall\, 
k\in\mathbb{N}.
\end{gather*}
For $n=3,4,5$ we prove that the WE algebra 
is isomorphic to the direct sum of $L(n)$ and a two-dimensional abelian Lie algebra. 
In particular, for $n=3$ the WE algebra of~\er{main} 
is isomorphic to the WE algebra 
of the anisotropic Landau-Lifshitz equation~\cite{ll}.
For $n=2$ the curve~\er{curve} is rational and system~\er{main} belongs 
to the well-studied class of scalar evolutionary equations~\cite{kdv1,kn}, so we skip the case $n=2$.  

To achieve this, we prove that the algebra $L(n)$ is isomorphic for any $n\ge 3$ 
to the Lie algebra given by generators $p_1,\dots,p_n$ and relations   
\begin{gather}
\label{rel1i}
[p_i,[p_i,p_k]]-[p_j,[p_j,p_k]]=(r_j-r_i)p_k\,\quad\text{for}\,\ 
i\neq k,\ j\neq k,\\
\notag
[p_i,[p_j,p_k]]=0\,\quad\text{for}\,\ 
i\neq k,\ i\neq j,\ j\neq k.
\end{gather}
In our opinion, this algebraic result may be of independent interest. 
For $n=3$ it was proved in~\cite{ll}. 

It is known that for a given evolutionary system 
\begin{equation}
\label{ut}
u^i_t=F^i(u^j,u^k_{x},u^l_{xx},\dots),\quad i=1,\dots,m, 
\end{equation} 
the WE algebra helps also to find a `modified' system 
\begin{equation}
\label{vt}
v^i_t=H^i(v^j,v^k_{x},v^l_{xx},\dots),\quad i=1,\dots,m,
\end{equation} 
connected with~\er{ut} by a B\"acklund transformation of Miura type 
(sometimes also called a differential substitution)
\begin{equation}
\label{mt}
u^i=g^i(v^j,v^k_{x},v^l_{xx},\dots),\quad i=1,\dots,m.
\end{equation} 
That is, for any solution $v^1,\dots,v^m$ of~\er{vt} functions~\er{mt} form a solution of~\er{ut}, 
and for any solution $u^1,\dots,u^m$ of~\er{ut} locally there exist functions $v^1,\dots,v^m$ 
satisfying~\er{vt},~\er{mt}. 

For system~\er{main} a modified system~\er{vt} and a transformation~\er{mt} 
were not known. Using a suitable vector field representation 
of the WE algebra, we find a family of systems~\er{vt} and 
transformations~\er{mt} for~\er{main}. This family is parametrized by points 
of the curve~\er{curve}.  
Apparently, this is the first example of a family of Miura type transformations 
parametrized by points of a curve of genus greater than 1.    
  
We would like to make also the following observation. 
Clearly, relations~\er{rel1i} look somewhat similar to equations~\er{curve}. 
And indeed, formula~\er{qie} and Theorem~\ref{gnL} below explain 
how the generator $p_i$ is related to $\la_i$. 
Note that, as we show in Sections~\ref{n3},~\ref{n4},~\ref{n5}, 
at least for $n=3,4,5$ relations~\er{rel1i} arise from the internal jet space geometry of system~\er{main} 
using the Wahlquist-Estabrook procedure. 
Therefore, the Wahlquist-Estabrook method and its generalization in~\cite{cfa,lie} 
may be useful also for the following problem: 
given a system of PDEs, which is suspected to be integrable, 
how to recover an algebraic curve that most naturally parametrizes a possible Lax pair for this system? 
For example, applying the Wahlquist-Estabrook procedure to system~\er{main} 
without any knowledge of Lax pairs or algebraic curves behind~\er{main}, 
one obtains relations~\er{rel1i} for the WE algebra. 
Since $p_i$ should correspond to some matrix-valued functions on an algebraic curve, 
looking at~\er{rel1i} it is not so hard to guess that the curve 
should be of the form~\er{curve}. 
It would be very interesting to make this observation into 
a rigorous construction for more PDEs.

Several more integrable systems associated 
with the curve~\er{curve} were introduced in~\cite{mll,skr,skr1,skr2}, 
where the opposite approach is taken: 
they start with an infinite-dimensional Lie algebra very similar 
to our algebra $L(n)$ and construct integrable systems from this Lie algebra.   
Note that a representation of the Lie algebra 
in terms of a finite number of generators 
and relations was not obtained in~\cite{mll,skr,skr1,skr2}.  

The functions $S=(s^1(x,t),\dots,s^n(x,t))$ in~\er{main} 
and the parameters $\la_i,\,r_i$ in~\er{curve} may take values in $\Com$ or $\mathbb{R}$. 
In this paper the $\Com$-valued case is studied, but all results and proofs are valid also 
in the $\mathbb{R}$-valued case, if one replaces $\Com$ by $\mathbb{R}$ 
in the definitions. 

The paper is organized as follows. 

In Section~\ref{wea} we present a rigorous definition of 
WE algebras for arbitrary evolutionary systems 
and, therefore, for any systems of PDEs that can be written in evolutionary form. 
For this we use formal power series with coefficients in Lie algebras. 
For example, such series occur in computations of the WE algebra of 
the $f$-Gordon equation $u_{tt}-u_{xx}=f(u)$~\cite{shadwick}, 
which, as is well known, can be rewritten in evolutionary form as follows
\begin{equation}
\label{evol}
\begin{aligned}
u_t&=q,\\
q_t&=u_{xx}+f(u).
\end{aligned}
\end{equation}
We discuss also possible generalizations of the Wahlquist-Estabrook ansatz. 

In Sections~\ref{comput},~\ref{repres} the above mentioned results 
on the WE algebra of system~\er{main} are obtained, and in Section~\ref{miura} 
the family of Miura type transformations is constructed. 
The appendix contains the proof of technical Lemma~\ref{lemma}. 

The following abbreviations are used in the paper: 
WE = Wahlquist-Estabrook, ZCR = zero-curvature representation. 

\section{The general definition of WE algebras} 
\label{wea}

Originally the Wahlquist-Estabrook prolongation method 
was formulated in terms of differential forms. 
We prefer the vector fields version of it, which goes as follows.  
For a given $m$-component evolutionary system of PDEs of order $d\ge 1$
\begin{equation}
\label{sys}
u^i_t=F^i(u^1,\dots,u^m,\,u^1_1,\dots,u^m_1,\dots,u^1_d,\dots,u^m_d),\quad
i=1,\dots,m,\quad u^i_k=\frac{\pd^k u^i}{\pd x^k},
\end{equation}
consider the infinite-dimensional jet space with the coordinates 
\begin{equation}
\label{coor}
x,\,t,\,u^i_k,\quad i=1,\dots,m,\quad k=0,1,2,\dots,\quad u^i_0=u^i.
\end{equation}
The total derivative operators 
\begin{equation*}
D_x=\frac{\pd}{\pd x}+\sum_{i,\,k}u^i_{k+1}\frac{\pd}{\pd u^i_k},\quad\quad 
D_t=\frac{\pd}{\pd t}+\sum_{i,\,k}D_x^{k+1}(F^i)\frac{\pd}{\pd u^i_k}
\end{equation*}
are commuting vector fields on this space. 

Let $\al$ be an $(m\times d)$-matrix with entries 
$$
\al_{ik}\in\zp,\quad i=1,\dots,m,\quad k=0,\dots,d-1.
$$ 
Denote the set of such matrices by $\mat$. 
Denote by $U^\al$ the following product of coordinates~\er{coor} 
$$
U^\al=\prod_{\substack{i=1,\dots,m,\\k=0,\dots,d-1}}\bigl(u^i_k\bigl)^{\alpha_{ik}}.
$$
Consider two formal power series 
\begin{equation}
\label{xt}
X=\sum_{\al\in \mat}A^\al U^\al,\quad T=\sum_{\beta\in \mat}B^\beta U^\beta.
\end{equation}
in the variables 
\begin{equation}
\label{var}
u^i_k,\quad 1\le i\le m,\quad 0\le k\le d-1,
\end{equation}
where the coefficients $A^\alpha,\,B^\beta$ are elements of some (not specified yet) Lie algebra. 
The equation 
\begin{equation}
\label{dxdt}
[D_x+X,D_t+T]=D_x(T)-D_t(X)+[X,T]=0
\end{equation} 
in the space of formal power series 
is equivalent to some Lie algebraic relations for the elements $A^\alpha,\,B^\beta$. 
Here 
\begin{gather*}
D_x(T)=\sum_{\beta\in \mat}B^\beta\cdot D_x(U^\beta),\quad D_t(X)=\sum_{\al\in \mat}A^\al\cdot D_t(U^\al),\\
[X,T]=\sum_{\al,\beta\in \mat} [A^\al, B^{\beta}]\cdot U^\al\cdot U^{\beta}.
\end{gather*}

Let $F$ be the free Lie algebra generated by all the letters $A^\alpha,\,B^\beta$ for $\al,\beta\in \mat$. 
The quotient of $F$ over the above mentioned relations 
arising from equation~\er{dxdt} is called the \emph{Wahlquist-Estabrook Lie algebra} 
of system~\er{sys} (or the \emph{WE algebra} in brief).  
From now on $A^\alpha,\,B^\beta$ are elements of the WE algebra. 
Then~\er{xt} is the most general solution of~\er{dxdt} provided that $X,\,T$ are 
power series in variables~\er{var}. 

\begin{remark} 
For many systems~\er{sys} equation~\er{dxdt} implies that $X,\,T$ are of the form 
\begin{equation}
\label{xtfg}
X=\sum_{i=1}^{k_1}C_if_i,\quad T=\sum_{j=1}^{k_2}D_jg_j,\quad k_1,k_2\in\mathbb{N}, 
\end{equation}
where $f_i,\,g_j$ are analytic functions of variables~\er{var}, 
the functions $f_1,\dots,f_{k_1}$ are linearly independent, 
the functions $g_1,\dots,g_{k_2}$ are linearly independent,  
and $C_i,\,D_j$ are elements of the WE algebra. 
Expanding $f_i,\,g_j$ as power series in~\er{var}, we obtain 
that $A^\alpha,\,B^\beta$ from~\er{xt} are linear combinations of $C_i,\,D_j$ 
and, therefore, in this case $C_i,\,D_j$ can be taken as another set of generators 
of the same WE algebra. 

However, the cases when $X,\,T$ are formal power series and cannot be presented 
as finite sums of analytic functions do also occur. 
For example, this happens for the $f$-Gordon equation $u_{tt}-u_{xx}=f(u)$~\cite{shadwick}, 
which can be rewritten in evolutionary form~\er{evol}.
\end{remark}

\begin{remark}
A natural question arises, what happens if in equation~\er{dxdt} one considers 
power series $X$, $T$ in the variables 
$$
u^i_k,\quad 1\le i\le m,\quad 0\le k\le s, 
$$ 
for arbitrary $s$? 
It turns out that if $s>d-1$ then before solving~\er{dxdt} one should simplify $X$, $T$  
by so-called gauge transformations~\cite{cfa,lie}, and then one also obtains 
certain Lie algebras~\cite{cfa,lie}, which are generally bigger than the WE algebra for $s=d-1$.  
For example, for the Krichever-Novikov equation the WE algebra for $s=d-1$ is trivial, 
but the case $s=d$ does produce an interesting Lie algebra~\cite{kn}. 

Coordinate-free meaning of the Lie algebras for arbitrary $s$ is studied in~\cite{cfa,lie}. 
In this way one obtains new geometric invariants of PDEs, which allow to prove, 
for example, that some PDEs are not connected by any B\"acklund transformations~\cite{cfa,lie}.
\end{remark}

Let $\mg$ be a Lie algebra. 
Recall~\cite{faddeev} 
that $\mg$-valued functions $M,\,N$ of a finite number of variables~\er{coor} form  
a \emph{zero-curvature representation} (ZCR in short) for system~\er{sys} if 
\begin{equation}
\label{mn}
[D_x+M,D_t+N]=D_x(N)-D_t(M)+[M,N]=0. 
\end{equation}
If $\dim\mg<\infty$ then $M,\,N$ are supposed to be analytic or smooth $\mg$-valued functions, 
while if $\dim\mg=\infty$ then $M,\,N$ are formal power series with coefficients in $\mg$. 
For example, the power series $X,\,T$ form a ZCR with 
values in the WE algebra. 

It is known that every Lax pair for a $(1+1)$-dimensional system of PDEs 
determines a ZCR. 
In addition to the Wahlquist-Estabrook technique, some other methods 
to obtain ZCRs for a given system of PDEs also exist 
(see, e.g.,~\cite{marvan,marvan2008,sakovich} and references therein). 

Since~\er{xt} is the most general solution of~\er{dxdt} 
and equation~\er{mn} is similar to~\er{dxdt}, 
we obtain the following result. 
\begin{proposition}
\label{wezcr}
Suppose that $M,\,N$ are $\mg$-valued functions of variables~\er{var} 
and form a ZCR. 
Expand $M,\,N$ as power series in~\er{var} 
\begin{equation*}
M=\sum_{\al\in \mat}M^\al U^\al,\quad N=\sum_{\beta\in \mat}N^\beta U^\beta,\quad 
M^\al,\,N^\beta\in\mg.
\end{equation*}
Then the map $A^\al\mapsto M^\al,\ B^\beta\mapsto N^\beta$ 
determines a homomorphism from the WE algebra to $\mg$. 
If the coefficients $M^\al,\,N^\beta,\,\al,\beta\in \mat$, generate 
the whole Lie algebra $\mg$ then this homomorphism is surjective. 
\end{proposition}

\section{Computations for the generalized Landau-Lifshitz system}
\label{comput}

In order to study the WE algebra of system~\er{main}, 
we need to resolve (locally) the constraint $\langle S,S\rangle=1$ 
for the vector $S=(s^1(x,t),\dots,s^n(x,t))$. 
Following~\cite{mll}, we do it as 
\begin{equation}
\label{sp}
s^i=\frac{2u^i}{1+\langle u,u\rangle},\quad
i=1,\dots,n-1,\quad
s^n=\frac{1-\langle u,u\rangle}{1+\langle u,u\rangle},
\end{equation}
where $u$ is an $(n-1)$-dimensional vector with the components $u^1(x,t),\dots,u^{n-1}(x,t)$, 
and $\langle\cdot,\cdot\rangle$ is the standard scalar product.
Then one can rewrite system~\er{main} as~\cite{mll}
\begin{multline}
\label{pt}
u_t=u_{xxx}-6\langle u,u_x\rangle\Delta^{-1}u_{xx}+
\bigl(-6\langle u,u_{xx}\rangle\Delta^{-1}+24\langle u,u_{x}\rangle^2\Delta^{-2}
-6\langle u,u\rangle\langle u_x,u_{x}\rangle\Delta^{-2}\bigl)u_x+\\
\bigl(6\langle u_x,u_{xx}\rangle\Delta^{-1}-12\langle u,u_x\rangle\langle u_x,u_{x}\rangle\Delta^{-2}\bigl)u
+\frac32\Bigl(r_n+4\Delta^{-2}\sum_{i=1}^{n-1}(r_i-r_n)(u^i)^2\Bigl)u_x,
\end{multline}
where $\Delta=1+\langle u,u\rangle$ and $r_1,\dots,r_n$ are distinct complex numbers 
that are the entries of the matrix $R=\mathrm{diag}(r_1,\dots,r_n)$ from system~\er{main}.
 
Let $E_{i,j}\in\mathfrak{gl}_{n+1}(\Com)$ be the matrix with 
$(i,j)$-th entry equal to 1 and all other entries equal to zero. 
Recall that the Lie subalgebra $\mathfrak{so}_{n,1}\subset\mathfrak{gl}_{n+1}(\Com)$ 
has the following basis
$$
E_{i,j}-E_{j,i},\quad i<j\le n,\quad E_{l,n+1}+E_{n+1,l},\quad l=1,\dots,n.
$$
From the results of~\cite{mll,skr-jmp} one can obtain the following 
$\mathfrak{so}_{n,1}$-valued ZCR of system~\er{main}
\begin{gather}
\label{M}
M=\sum_{i=1}^ns^i\la_i(E_{i,n+1}+E_{n+1,i}),\\
\label{N}
N=D_x^2(M)+[D_x(M),M]+(r_1+\la_1^2)M+\bigl(\frac12(S,RS)+\frac32(S_x,S_x)\bigl)M.
\end{gather}
Here $\la_1,\dots,\la_n$ are complex parameters satisfying equations~\er{curve}. 
\begin{remark}
It was noticed in~\cite{skr} that the formulas
$\la=\la_i^2+r_i,\ y=\prod_{i=1}^n\la_i,$ 
provide a mapping from the curve~\er{curve} 
to the hyperelliptic curve $y^2=\prod_{i=1}^n(\la-r_i)$. 
However, according to~\cite{mll}, 
the curve~\er{curve} itself is not hyperelliptic.
\end{remark}
If $S=(s^1,\dots,s^n)$ is given by formulas~\er{sp} then~\er{M},~\er{N} 
determines a ZCR for system~\er{pt}. 

In the algebra $\Com[\la_1,\dots,\la_n]$
consider the ideal $I\subset\Com[\la_1,\dots,\la_n]$ 
generated by the polynomials $\la_i^2-\la_j^2+r_i-r_j,\,\ 
i,j=1,\dots,n$.  
Denote by $\bar\la_i$ the image of $\la_i$ in the quotient 
algebra $Q=\Com[\la_1,\dots,\la_n]/I$, which is equal to the 
algebra of polynomial functions on the curve~\er{curve}. 

Consider the infinite-dimensional Lie algebra over $\Com$
\begin{gather*}
\mathfrak{gl}_{n+1}(\Com)\otimes_\Com Q\cong\mathfrak{gl}_{n+1}(Q),\\ 
[g_1\otimes q_1,\,g_2\otimes q_2]=[g_1,g_2]\otimes q_1q_2,\quad 
g_i\in \mathfrak{gl}_{n+1}(\Com),\quad q_i\in Q,
\end{gather*}
and the elements 
\begin{equation}
\label{qie}
Q_i=(E_{i,n+1}+E_{n+1,i})\otimes\bar\la_i\,\in\,
\mathfrak{so}_{n,1}\otimes_\Com Q\subset 
\mathfrak{gl}_{n+1}(\Com)\otimes_\Com Q,\quad i=1,\dots,n.
\end{equation}
Denote by $L(n)\subset\mathfrak{so}_{n,1}\otimes Q$ the Lie subalgebra  
generated by~$Q_1,\dots,Q_n$. 

Obviously, the element $\bar\la=\bar\la_i^2+r_i\in Q$ does not depend on $i$.  
For $i,j=1,\dots,n$, $i\neq j$, and $k\in\mathbb{N}$ 
consider the following elements of $\mathfrak{so}_{n,1}\otimes Q$  
\begin{gather*}
Q^{2k-1}_i=(E_{i,n+1}+E_{n+1,i})\otimes\bar\la^{k-1}\bar\la_i,\quad
Q^{2k}_{ij}=(E_{i,j}-E_{j,i})\otimes\bar\la^{k-1}\bar\la_i\bar\la_j. 
\end{gather*} 
For $i,j,a,b=1,\dots,n$, $i\neq j$, $a\neq b$, and $k_1,k_2\in\mathbb{N}$ one has 
\begin{multline} 
\label{q1}
[Q^{2k_1}_{ij},\,Q^{2k_2}_{ab}]=\delta_{aj}Q^{2(k_1+k_2)}_{ib} 
-\delta_{ib}Q^{2(k_1+k_2)}_{aj}
+\delta_{jb}Q^{2(k_1+k_2)}_{ai}-\delta_{ia}Q^{2(k_1+k_2)}_{jb}\\
+r_i\delta_{ib}Q^{2(k_1+k_2-1)}_{aj}
-r_j\delta_{aj}Q^{2(k_1+k_2-1)}_{ib}+r_i\delta_{ia}Q^{2(k_1+k_2-1)}_{jb} 
-r_j\delta_{jb}Q^{2(k_1+k_2-1)}_{ai},  
\end{multline}
\begin{equation} 
\label{q2}
[Q^{2k_1}_{ij},\,Q^{2k_2-1}_{a}]=\delta_{aj}Q^{2k_1+2k_2-1}_{i}
-\delta_{ia}Q^{2k_1+2k_2-1}_{j}-r_j\delta_{aj}Q^{2k_1+2k_2-3}_{i}+r_i\delta_{ia}Q^{2k_1+2k_2-3}_{j},
\end{equation}
\begin{equation} 
\label{q3}
[Q^{2k_1-1}_{i},\,Q^{2k_2-1}_{j}]=Q^{2(k_1+k_2-1)}_{ij},\quad 
[Q^{2k_1-1}_{i},\,Q^{2k_2-1}_{i}]=0.   
\end{equation}

Since $Q^1_i=Q_i$ and $Q^{2k}_{ij}=-Q^{2k}_{ji}$, 
from~\er{q1},~\er{q2},~\er{q3} we obtain that 
$$
Q^{2k-1}_l,\quad Q^{2k}_{ij},\ \quad i,j,l=1,\dots,n,\quad 
i<j,\quad k=1,2,3,\dots,
$$
span the Lie algebra $L(n)$. 
It is easily seen that these elements are linearly independent over $\Com$
and, therefore, form a basis of $L(n)$. 

For $k\in\mathbb{N}$ set  
$$
\bar L_{2k-1}=\langle Q^{2k-1}_l\,|\,l=1,\dots,n\rangle, 
\quad \bar L_{2k}=\langle Q^{2k}_{ij}\,|\,i,j=1,\dots,n,\ i<j\rangle. 
$$ 
Here and below for elements $v_1,\dots,v_s$ of a vector space 
the expression $\langle v_1,\dots,v_s\rangle$ denotes the linear span 
of $v_1,\dots,v_s$ over $\Com$. Then from~\er{q1}, \er{q2}, \er{q3} one gets 
$$
L(n)=\bigoplus_{i=1}^\infty\bar L_i,\quad\  
[\bar L_i,\bar L_j]\subset \bar L_{i+j}+\bar L_{i+j-2}. 
$$
Therefore, the Lie algebra $L(n)$ is quasigraded~\cite{quasigraded}. 

Clearly, formulas~\er{M},~\er{N} can be regarded as a ZCR 
with values in the algebra $L(n)$. 
In particular, $M=\sum_{i=1}^n s^iQ_i$. 
Combining this with Proposition~\ref{wezcr}, we obtain the following. 
\begin{theorem}
For any $n\ge 3$ 
we have an epimorphism from the WE algebra of system~\er{main} 
onto the infinite-dimensional Lie algebra $L(n)$.
\end{theorem}
Let us give a complete description of the WE algebra of~\er{pt} 
for small $n$. For $n=2$ system~\er{pt} is a scalar equation of the form 
$u_t=u_{xxx}+f(u,u_x,u_{xx})$. Since for such equations the WE algebras 
have already been studied quite extensively~(see, e.g, \cite{kdv1,kn} 
and references therein) and the curve~\er{curve} is rational for $n=2$, 
we skip the case $n=2$. 
For $n=3,4,5$ the WE algebras are studied below. 
\subsection{The case n=3}
\label{n3}
According to Section~\ref{wea}, we must solve equation~\er{dxdt} for
$$
X=X(u^1,u^2,u^1_x,u^2_x,u^1_{xx},u^2_{xx}),\quad 
T=T(u^1,u^2,u^1_x,u^2_x,u^1_{xx},u^2_{xx}).
$$
If we differentiate~\er{dxdt} with respect to the variables 
$u^1_{xxxxx}$, $u^2_{xxxxx}$, $u^1_{xxxx}$, $u^2_{xxxx}$, 
we get that $X$ depends only on $(u^1,u^2)$. 
Next, differentiating~\er{dxdt} with respect to $u^1_{xxx}$, $u^2_{xxx}$, $u^1_{xx}$, $u^2_{xx}$ several times, 
one obtains 
\begin{equation*}
X=X(u^1,u^2), \quad T=\frac{\pd X}{\pd u^1} u^1_{xx}+\frac{\pd X}{\pd u^2} u^2_{xx}+F_1(u^1,u^2,u^1_x,u^2_x), 
\end{equation*}
where
\begin{multline*}
F_{1} =-\frac12\left(\frac{\pd^2 X}{\pd u^1\pd u^1}\Delta+6u^1\frac{\pd X}{\pd u^1}-6u^2\frac{\pd X}{\pd u^2}\right)\Delta^{-1}(u^1_{x})^{2}\\
-\left(\frac{\pd^2 X}{\pd u^1\pd u^2}\Delta+6u^2\frac{\pd X}{\pd u^1}+6u^1\frac{\pd X}{\pd u^2}\right)\Delta^{-1}u^1_{x}u^2_{x}\\
-\frac12\left(\frac{\pd^2 X}{\pd u^2\pd u^2}\Delta+6u^2\frac{\pd X}{\pd u^2}-6u^1\frac{\pd X}{\pd u^1}\right)\Delta^{-1}(u^2_{x})^{2}
+[\frac{\pd X}{\pd u^1},X]u^1_{x}+[\frac{\pd X}{\pd u^2},X]u^2_{x}+F_{2}(u^1,u^2),
\end{multline*}
and $\Delta=1+(u^1)^2+(u^2)^2$. 

Then the expression 
\begin{equation}
\label{E}
[D_x+X,D_t+T]=D_x(T)-D_t(X)+[X,T]
\end{equation}
becomes a third degree polynomial in $u^1_{x},\,u^2_{x}$ with coefficients depending on $u^1,\,u^2$. 
By putting equal to zero the third degree coefficients, 
we get a system of linear PDEs of order 3 for the function $X(u^1,u^2)$. 
The general solution of this system is
\begin{equation}
\label{Xd}
X=\Bigl(D_0+u^1D_1+u^2D_2+\bigl((u^1)^2+(u^2)^2\bigl)D_3\Bigl)\Delta^{-1}, 
\end{equation}
where $D_i$ do not depend on $u^1,\,u^2$ and, therefore, belong to a set of generators of the WE algebra. 

Taking into account, that system~\er{pt} arises from~\er{main} by means of~\er{sp}, 
it is more convenient to rewrite~\er{Xd} as  
\begin{equation*}
X=\Bigl(2C_1u^1+2C_2u^2+C_3\bigl(1-(u^1)^2-(u^2)^2\bigl)\Bigl)\Delta^{-1}+C_0,
\end{equation*}
where $C_i$ are also elements of the WE algebra. 
By putting to zero the coefficients of~\er{E} at the monomials $(u^1_x)^2$ and $(u^2_x)^2$, 
we obtain 
\begin{equation}\label{eq.relazioni.preliminari}
[C_0,C_1]=[C_0,C_2]=[C_0,C_3]=0.
\end{equation}
Now, by putting the coefficients of $u^1_x$ and $u^2_x$ equal to zero,
one gets a system of two first order differential equations for $F_2(u^1,u^2)$. 
Its compatibility conditions are satisfied if and
only if the following relations hold 
\begin{gather}
\label{c0cij}
[C_0,[C_1,C_2]]=[C_0,[C_1,C_3]]=[C_0,[C_2,C_3]]=0,\\
\label{cijk}
[C_1,[C_2,C_3]]=[C_2,[C_3,C_1]]=[C_3,[C_1,C_2]]=0,\\
\label{113}
[C_1,[C_1,C_3]]-[C_2,[C_2,C_3]]=(r_2-r_1)C_3,\\
\label{112}
[C_1,[C_1,C_2]]-[C_3,[C_3,C_2]]=(r_3-r_1)C_2,\\
\label{221}
[C_2,[C_2,C_1]]-[C_3,[C_3,C_1]]=(r_3-r_2)C_1.
\end{gather}
Note that~\er{c0cij} follows from \eqref{eq.relazioni.preliminari}
using the Jacobi identity.

Then we are able to compute $F_2(u^1,u^2)$. 
\begin{multline*}
F_2=\Big((-2[C_3,[C_1,C_3]]+3C_1r_3)(u^1)^5+(3C_2r_3-2[C_3,[C_2,C_3]])(u^1)^4u^2\\
+(6C_1r_3-4[C_3,[C_1,C_3]])(u^1)^3(u^2)^2+(-4[C_3,[C_2,C_3]]+6C_2r_3)(u^1)^2(u^2)^3\\
+(-2[C_3,[C_1,C_3]]+3C_1r_3)u^1(u^2)^4+(3C_2r_3-2[C_3,[C_2,C_3]])(u^2)^5\\
+(2C_3r_2+2[C_2,[C_2,C_3]]-2C_3r_1+3C_3r_3)(u^1)^4\\
+(4[C_2,[C_2,C_3]]+2C_3r_2-2C_3r_1+6C_3r_3)(u^1)^2(u^2)^2\\
+(2[C_2,[C_2,C_3]]+3C_3r_3)(u^2)^4+(4C_1r_1-4[C_3,[C_1,C_3]]+2C_1r_3)(u^1)^3\\
+(2C_2r_3-4[C_3,[C_2,C_3]]+4C_2r_1)(u^1)^2u^2+(-4[C_3,[C_1,C_3]]+4C_1r_2+2C_1r_3)u^1(u^2)^2\\
+(4C_2r_2+2C_2r_3-4[C_3,[C_2,C_3]])(u^2)^3+(4[C_2,[C_2,C_3]]+4C_3r_2+2C_3r_1)(u^1)^2\\
+(4[C_2,[C_2,C_3]]+6C_3r_2)(u^2)^2+(-2[C_3,[C_1,C_3]]+3C_1r_3)u^1\\
+(3C_2r_3-2[C_3,[C_2,C_3]])u^2+C_3r_3+2C_3r_2+2[C_2,[C_2,C_3]]\Big)\Delta^{-3}+C',
\end{multline*}
where $C'$ is another generator of the WE algebra.

Then~\er{E} depends only on $u^1, u^2$ and vanishes if and only if 
\begin{gather}
\label{c3c'}
[C_3,C']=0,\\
\label{c0c'}
[C_0,C']=0,\\ 
\label{c3223}
[C_3,[C_2,[C_2,C_3]]]=0, \\
\label{c3313}
[C_3,[C_3,[C_3,C_1]]]=\frac32 r_3[C_1,C_3]+[C_1,C'],\\
\label{c3323}
[C_3,[C_3,[C_3,C_2]]]=\frac32 r_3[C_2,C_3]+[C_2,C'],\\
\label{c2312}
[[C_2,C_3],[C_1,C_2]]=\frac12 r_3[C_1,C_3]+[C_1,C']+r_2[C_1,C_3],\\
\label{c2223}
[C_2,[C_2,[C_2,C_3]]]=-\frac12 r_3[C_2,C_3]-[C_2,C']-r_2[C_2,C_3].
\end{gather}

Therefore, the WE algebra of~\er{pt} for $n=3$ is given by the generators 
$C'$, $C_0$, $C_1$, $C_2$, $C_3$ and all relations obtained in this subsection. 
Let us simplify the structure of these relations. 
Taking into account~\er{eq.relazioni.preliminari} and~\er{cijk}, 
relations~\er{c3c'},~\er{c0c'},~\er{c3223},~\er{c2312},~\er{c2223} imply 
\begin{equation}
\label{c'i}
[C'+\big(r_2+\frac12 r_3\big)C_3+[C_2,[C_2,C_3]],\,C_i]=0,\quad i=0,1,2,3.
\end{equation}
Then it is easily seen that all relations follow from~\er{eq.relazioni.preliminari}, 
\er{cijk}, \er{113}, \er{112}, \er{221}, \er{c'i} using the Jacobi identity. 
For example, let us prove this for~\er{c3223} and~\er{c3313}, 
for all other relations the statement can be proved analogously.   
Applying $\mathrm{ad}\,C_3$ to~\er{113}, $\mathrm{ad}\,C_2$ to~\er{112}, and $\mathrm{ad}\,C_1$ to~\er{221}, 
one obtains 
\begin{equation*}
[C_i,[C_j,[C_j,C_i]]]=[C_i,[C_k,[C_k,C_i]]]\quad\text{for}\quad\{i,j,k\}=\{1,2,3\}.
\end{equation*}
Besides, by the Jacobi identity, we have $[C_i,[C_j,[C_j,C_i]]]=-[C_j,[C_i,[C_i,C_j]]]$. 
Therefore,
\begin{multline*}
[C_3,[C_2,[C_2,C_3]]]=-[C_2,[C_3,[C_3,C_2]]]=-[C_2,[C_1,[C_1,C_2]]]=\\
[C_1,[C_2,[C_2,C_1]]]=[C_1,[C_3,[C_3,C_1]]]=-[C_3,[C_1,[C_1,C_3]]]=-[C_3,[C_2,[C_2,C_3]]],
\end{multline*}
which implies~\er{c3223}.

Using~\er{221},~\er{cijk}, and the Jacobi identity, one obtains 
\begin{multline*}
[C_3,[C_3,[C_3,C_1]]]=[C_3,[C_2,[C_2,C_1]]]+(r_2-r_3)[C_3,C_1]=\\ 
[[C_3,C_2],[C_2,C_1]]+(r_2-r_3)[C_3,C_1]=-[C_1,[C_2,[C_2,C_3]]]+(r_2-r_3)[C_3,C_1].
\end{multline*}
Combining this with~\er{c'i} for $i=1$, we get~\er{c3313}.

Thus the WE algebra is isomorphic to the direct sum $\mg(3)\oplus A$, where $\mg(3)$ is the Lie algebra given 
by the generators $C_1,\,C_2,\,C_3$ and relations~\er{cijk},~\er{113},~\er{112},~\er{221}, 
and $A$ is the two-dimensional abelian Lie algebra spanned by $C_0$, $C'+\big(r_2+\dfrac12 r_3\big)C_3+[C_2,[C_2,C_3]]$.  

\subsection{The case n=4}
\label{n4}
According to Section~\ref{wea}, we must solve equation~\er{dxdt} for
$$
X=X(u^1,u^2,u^3,u^1_x,u^2_x,u^3_x,u^1_{xx},u^2_{xx},u^3_{xx}),\quad 
T=T(u^1,u^2,u^3,u^1_x,u^2_x,u^3_x,u^1_{xx},u^2_{xx},u^3_{xx}).
$$
Similarly to the previous subsection we obtain the following.
$$
X=\Big(2C_1u^1+2C_2u^2+2C_3u^3+C_4\big(1-(u^1)^2-(u^2)^2-(u^3)^2\big)\Big)\Delta^{-1}+C_0.
$$
\begin{multline*}
T=\sum_{i=1}^3 \frac{\pd X}{\pd u^i} u^i_{xx}-\frac12\sum_{i=1}^{3}\left(\frac{\pd^2 X}{\pd u^i\pd u^i}\Delta+6\Big(\frac{\pd X}{\pd u^i}u^i-\sum_{j\neq i}\frac{\pd X}{\pd u^j}u^j\Big)\right)\Delta^{-1}
(u^i_{x})^2\\
-\sum_{1\le i<j\le 3}\left(\frac{\pd^2 X}{\pd u^i\pd u^j}\Delta+6\Big(\frac{\pd X}{\pd u^i}u^j+\frac{\pd X}{\pd u^j}u^i\Big)\right)\Delta^{-1}u_x^iu_x^j
+\sum_{i=1}^3[\frac{\pd X}{\pd u^i},X]u_{x}^{i}+F_{2}(u^1,u^2,u^3), 
\end{multline*}
where $\Delta=1+(u^1)^2+(u^2)^2+(u^3)^2$ and 
\begin{multline}
\label{F_2n4}
F_2=\Big((-2[C_4,[C_1,C_4]]+3C_1r_4)(u^1)^5+(3C_2r_4-2[C_4,[C_2,C_4]])(u^1)^4u^2\\
+(3C_3r_4-2[C_4,[C_3,C_4]])(u^1)^4u^3+(-4[C_4,[C_1,C_4]]+6C_1r_4)(u^1)^3(u^2)^2\\
+(-4[C_4,[C_1,C_4]]+6C_1r_4)(u^1)^3(u^3)^2+(-4[C_4,[C_2,C_4]]+6C_2r_4)(u^1)^2(u^2)^3\\
+(6C_3r_4-4[C_4,[C_3,C_4]])(u^1)^2(u^2)^2u^3+(-4[C_4,[C_2,C_4]]+6C_2r_4)(u^1)^2u^2(u^3)^2\\
+(6C_3r_4-4[C_4,[C_3,C_4]])(u^1)^2(u^3)^3+(-2[C_4,[C_1,C_4]]+3C_1r_4)u^1(u^2)^4\\
+(-4[C_4,[C_1,C_4]]+6C_1r_4)u^1(u^2)^2(u^3)^2+(-2[C_4,[C_1,C_4]]+3C_1r_4)u^1(u^3)^4\\
+(3C_2r_4-2[C_4,[C_2,C_4]])(u^2)^5+(3C_3r_4-2[C_4,[C_3,C_4]])(u^2)^4u^3\\
+(-4[C_4,[C_2,C_4]]+6C_2r_4)(u^2)^3(u^3)^2+(6C_3r_4-4[C_4,[C_3,C_4]])(u^2)^2(u^3)^3\\
+(3C_2r_4-2[C_4,[C_2,C_4]])u^2(u^3)^4+(3C_3r_4-2[C_4,[C_3,C_4]])(u^3)^5\\
+(2C_4r_3+2[C_3,[C_3,C_4]]+3C_4r_4-2C_4r_1)(u^1)^4\\
+(4[C_3,[C_3,C_4]]-2C_4r_2-2C_4r_1+6C_4r_4+4C_4r_3)(u^1)^2(u^2)^2\\
+(4[C_3,[C_3,C_4]]+2C_4r_3+6C_4r_4-2C_4r_1)(u^1)^2(u^3)^2\\
+(-2C_4r_2+3C_4r_4+2[C_3,[C_3,C_4]]+2C_4r_3)(u^2)^4\\
+(-2C_4r_2+4[C_3,[C_3,C_4]]+2C_4r_3+6C_4r_4)(u^2)^2(u^3)^2+(3C_4r_4+2[C_3,[C_3,C_4]])(u^3)^4\\
+(2C_1r_4-4[C_4,[C_1,C_4]]+4C_1r_1)(u^1)^3+(-4[C_4,[C_2,C_4]]+2C_2r_4+4C_2r_1)(u^1)^2u^2\\
+(-4[C_4,[C_3,C_4]]+4C_3r_1+2C_3r_4)(u^1)^2u^3+(-4[C_4,[C_1,C_4]]+2C_1r_4+4C_1r_2)u^1(u^2)^2\\
+(2C_1r_4-4[C_4,[C_1,C_4]]+4C_1r_3)u^1(u^3)^2+(-4[C_4,[C_2,C_4]]+4C_2r_2+2C_2r_4)(u^2)^3\\
+(4C_3r_2+2C_3r_4-4[C_4,[C_3,C_4]])(u^2)^2u^3+(-4[C_4,[C_2,C_4]]+4C_2r_3+2C_2r_4)u^2(u^3)^2\\
+(-4[C_4,[C_3,C_4]]+4C_3r_3+2C_3r_4)(u^3)^3+(2C_4r_1+4C_4r_3+4[C_3,[C_3,C_4]])(u^1)^2\\
+(4[C_3,[C_3,C_4]]+2C_4r_2+4C_4r_3)(u^2)^2+(4[C_3,[C_3,C_4]]+6C_4r_3)(u^3)^2\\
+(-2[C_4,[C_1,C_4]]+3C_1r_4)u^1+(3C_2r_4-2[C_4,[C_2,C_4]])u^2\\
+(3C_3r_4-2[C_4,[C_3,C_4]])u^3+2C_4r_3+2[C_3,[C_3,C_4]]+C_4r_4\Big)\Delta^{-3}+C'.
\end{multline}
Here $C',\,C_0,\,C_1,\,C_2,\,C_3,\,C_4$ are the generators of the WE algebra. 

The WE algebra is equal to the direct sum $\mg(4)\oplus A$, where $\mg(4)$ is the Lie algebra given 
by the generators $C_1,\,C_2,\,C_3,\,C_4$ and the relations
\begin{gather}
\label{rel1-n4}
[C_i,[C_i,C_k]]-[C_j,[C_j,C_k]]=(r_j-r_i)C_k\,\quad\text{for}\,\ 
i\neq k,\ j\neq k,\\
\label{rel2-n4}
[C_i,[C_j,C_k]]=0\,\quad\text{for}\,\ 
i\neq k,\ i\neq j,\ j\neq k,
\end{gather}
and $A$ is the two-dimensional abelian Lie algebra spanned by 
$$
C_0,\quad C'+\big(r_3+\dfrac12 r_4\big)C_4+[C_3,[C_3,C_4]].
$$
In particular, one has
$$
[C_i,\,C_0]=[C_i,\,C'+\big(r_3+\dfrac12 r_4\big)C_4+[C_3,[C_3,C_4]]]=0\quad\text{for}\quad i=0,1,2,3,4.
$$

\subsection{The case n=5}
\label{n5}
Similarly to the previous subsections one gets the following results.
$$
X=\Big(2C_1u^1+2C_2u^2+2C_3u^3+2C_4u^4+C_5\big(1-(u^1)^2-(u^2)^2-(u^3)^2-(u^4)^2\big)\Big)\Delta^{-1}+C_0.
$$
\begin{multline*}
T=\sum_{i=1}^4 \frac{\pd X}{\pd u^i} u^i_{xx}-\frac12\sum_{i=1}^4\left(\frac{\pd^2 X}{\pd u^i\pd u^i}\Delta+6\Big(\frac{\pd X}{\pd u^i}u^i-\sum_{j\neq i}\frac{\pd X}{\pd u^j}u^j\Big)\right)\Delta^{-1}
(u^i_{x})^2\\
-\sum_{1\le i<j\le 4}\left(\frac{\pd^2 X}{\pd u^i\pd u^j}\Delta+6\Big(\frac{\pd X}{\pd u^i}u^j+\frac{\pd X}{\pd u^j}u^i\Big)\right)\Delta^{-1}u_x^iu_x^j
+\sum_{i=1}^4[\frac{\pd X}{\pd u^i},X]u_{x}^{i}+F_{2}(u^1,u^2,u^3,u^4). 
\end{multline*}
Here $\Delta=1+(u^1)^2+(u^2)^2+(u^3)^2+(u^4)^2$ and 
\begin{equation*}
F_2=P(u^1,u^2,u^3,u^4)\Delta^{-3}+C',
\end{equation*}
where $P(u^1,u^2,u^3,u^4)$ is a cumbersome polynomial of degree 5 similar to~\er{F_2n4} 
and $C'$, $C_0$, $C_1$, $C_2$, $C_3$, $C_4$, $C_5$ are the generators of the WE algebra. 

The WE algebra is equal to the direct sum $\mg(5)\oplus A$, where $\mg(5)$ is the Lie algebra given 
by the generators $C_1,\,C_2,\,C_3,\,C_4,\,C_5$ and relations~\er{rel1-n4},~\er{rel2-n4} for $i,j,k=1,\dots,5$, 
and $A$ is the two-dimensional abelian Lie algebra spanned by 
$$
C_0,\quad C'+\big(r_4+\dfrac12 r_5\big)C_5+[C_4,[C_4,C_5]].
$$
In particular, we have
$$
[C_i,\,C_0]=[C_i,\,C'+\big(r_4+\dfrac12 r_5\big)C_5+[C_4,[C_4,C_5]]]=0\quad\text{for}\quad i=0,1,\dots,5.
$$

The explicit structure of the obtained algebras is described in the next section.

\section{Explicit structure of the algebras}
\label{repres}

Let $F$ be the free Lie algebra generated 
by the letters $p_1,\dots,p_n$, $n\ge 3$. 
Denote by $\mg(n)$ the quotient of $F$ over the relations 
\begin{gather}
\label{rel1}
[p_i,[p_i,p_k]]-[p_j,[p_j,p_k]]=(r_j-r_i)p_k\,\quad\text{for}\,\ 
i\neq k,\ j\neq k,\\
\label{rel2}
[p_i,[p_j,p_k]]=0\,\quad\text{for}\,\ 
i\neq k,\ i\neq j,\ j\neq k.
\end{gather}
The natural image of $p_i\in F$ in the quotient Lie algebra $\mg(n)$ 
is denoted by the same symbol $p_i$. 

The results of Sections~\ref{n3},~\ref{n4},~\ref{n5} can be summarized as follows. 
\begin{theorem}
\label{n345}
For $n=3,4,5$ the WE algebra of system~\er{pt} is isomorphic to the direct sum 
of $\mg(n)$ and a two-dimensional abelian Lie algebra. 
\end{theorem}
The next theorem describes the explicit structure of $\mg(n)$. 
\begin{theorem} 
\label{gnL}
For any $n\ge 3$ the mapping 
\begin{equation}
\label{is}
\vf\cl\mg(n)\to L(n),\quad p_i\mapsto Q_i, 
\end{equation}
is an isomorphism of Lie algebras. 
Here $L(n)\subset\mathfrak{so}_{n,1}\otimes Q$ is the infinite-dimensional 
Lie algebra defined in Section~\ref{comput}.
\end{theorem}
\begin{proof}
For $n=3$ this theorem was proved in~\cite{ll} 
for a different matrix representation of $L(n)$,  
and we will use a similar method. 

It is easy to check that $Q_i$ satisfy 
relations~\eqref{rel1},~\eqref{rel2},  
so mapping~\eqref{is} is an epimorphism. 

Define a filtration on $L(n)$ by vector subspaces $L^m$ as follows 
\begin{gather*}
L^0=0,\quad
L^1=\langle Q_1,\dots,Q_n\rangle,\quad 
L^m=L^1+\sum_{i+j\le m}[L^i,L^j]\quad\text{for}\ m>1,\\ 
L^0 \subset L^1 \subset L^2 \subset \cdots\quad
L(n)=\bigcup_i L^i. 
\end{gather*}

From~\er{q1},~\er{q2},~\er{q3} one gets that for all $k\in\mathbb{N}$ 
\begin{gather*}
Q^{2s-1}_l,\ Q^{2s}_{ij},\quad i,j,l\!=\!1,\dots,n,\,\ i<j,\,\ s\le k, 
\ \text{form a basis of }L^{2k},\\
Q^{2s-1}_l,\ Q^{2s-2}_{ij},\quad i,j,l\!=\!1,\dots,n,\,\ i<j,\,\ s\le k, 
\ \text{form a basis of }L^{2k-1}. 
\end{gather*}
This implies 
\begin{equation}
\label{L/L}
\dim(L^m/L^{m-1})=
\left\{
\begin{array}{c}
n,\quad\text{if $m$ is odd},\\
{n(n-1)}/{2},\quad\text{if $m$ is even}.
\end{array}
\right.
\end{equation}

Consider the similar filtration on $\mg(n)$ by vector subspaces $\mg^m$ 
$$
\mg^0=0,\quad
\mg^1=\langle p_1,\dots,p_n\rangle,\quad 
\mg^m=\mg^1+\sum_{i+j\le m}[\mg^i,\mg^j]\quad\text{for}\ m>1. 
$$
Clearly, 
\begin{equation}
\label{gL}
\vf(\mg^m)=L^m.
\end{equation} 
Combining this with~\er{L/L}, we obtain that it remains to prove 
\begin{equation}
\label{g/g}
\dim(\mg^m/\mg^{m-1})\le
\left\{
\begin{array}{c}
n,\quad\text{if $m$ is odd},\\
{n(n-1)}/{2},\quad\text{if $m$ is even}.
\end{array}
\right.
\end{equation} 
Indeed, if~\er{g/g} holds then~\er{L/L} and~\er{gL} 
imply that $\vf$ is an isomorphism. 

For $n=3$ statement~\er{g/g} was proved in~\cite{ll}. 
Below we suppose $n\ge 4$. 

For each $k\in\mathbb{N}$ set 
\begin{gather*}
P^{2k}_{ij}=(\ad p_i)^{2k-1}(p_j),\quad\text{for}\quad i\neq j,\\
P^{2k-1}_{1}=(\ad p_2)^{2k-2}(p_1),\\
P^{2k-1}_{i}=(\ad p_1)^{2k-2}(p_i)\quad\text{for}\quad i=2,3,\dots,n.
\end{gather*}

We will use the following notation for iterated commutators 
$$
[a_1\,a_2\,\dots a_{k-1}\,a_k]=[a_1,[a_2,[\dots,[a_{k-1},a_k]]\dots].
$$
Also, for brevity we replace each generator $p_i$ by the corresponding index $i$. 
So, for example, 
\begin{gather*}
[ii[jjk]lk]=[p_i,[p_i,[[p_j,[p_j,p_k]],[p_l,p_k]]]],\\ 
P^{2k}_{ij}=[\ub{i\dots i}_{2k-1}j],\quad 
P^{2k-1}_{1}=[\ub{2\dots 2}_{2k-2}1],\\ 
P^{2k-1}_{i}=[\ub{1\dots 1}_{2k-2}i]\quad\text{for}\quad i=2,3,\dots,n. 
\end{gather*}
For such an iterated commutator $C$ of several $p_i$ denote 
by $\mathrm{o}(C)$ the number of $p_i$ that appear in $C$. 
For example, 
\begin{gather*}
\mathrm{o}([ij])=2,\quad 
\mathrm{o}([iij])=3,\quad \mathrm{o}([ii[jjk]lk])=7,\\
\mathrm{o}(P^{2k}_{ij})=2k,\quad \mathrm{o}(P^{2k-1}_{i})=2k-1,\quad 
\mathrm{o}([P^{2k_1}_{ij},P^{2k_2}_{i'j'}])=2(k_1+k_2). 
\end{gather*} 
For two such iterated commutators $A,\,B$ we write $A\simeq B$ if 
$$
A-B\in\mg^m\quad\text{for }\ 
m=\max(\mathrm{o}(A),\mathrm{o}(B))-1. 
$$
In particular, $A\simeq 0$ means that $A\in\mg^m$ for $m=\mathrm{o}(A)-1$. 
\begin{remark}
Denote by $\bar C$ the image of $C$ in the quotient vector space $\mg^{\mathrm{o}(C)}/\mg^{\mathrm{o}(C)-1}$. 
Since $[\mg^i,\mg^j]\subset\mg^{i+j}$, we can consider the associated graded Lie algebra 
$$
\mathrm{gr}(\mg)=\bigoplus_i\mg^{i}/\mg^{i-1}.
$$
So, $A\simeq B$ if and only if $\bar A=\bar B$ in $\mathrm{gr}(\mg)$. 
In Lemma~\ref{lemma} below we essentially obtain some identities in the algebra $\mathrm{gr}(\mg)$. 
\end{remark}

From Lemma~\ref{lemma} by induction on $k$ one obtains that 
\begin{gather*}
P^{2s-1}_l,\quad P^{2s}_{ij},\quad i,j,l=1,\dots,n,\quad i<j,\quad s\le k, 
\quad\text{span }\mg^{2k},\\
P^{2s-1}_l,\quad P^{2s-2}_{ij},\quad i,j,l=1,\dots,n,\quad i<j,\quad s\le k, 
\quad\text{span }\mg^{2k-1}, 
\end{gather*}
which implies~\er{g/g}.
This lemma is proved in the Appendix.

\begin{lemma} 
\label{lemma}
If $i,j,i',j'$ are distinct integers from $\{1,\dots,n\}$ then 
for all $k_1,\,k_2\in\mathbb{Z}_+$ one has 
\begin{gather}
\label{PP0} 
[[\ub{i\dots i}_{2k_1}j][\ub{i\dots i}_{2k_2}j]]\simeq 0,\quad 
\text{in particular},\ 
[P^{2k_1+1}_j,P^{2k_2+1}_j]\simeq 0,\\ 
\label{PP1}
P^{2(k_1+k_2+1)}_{ij}\simeq -P^{2(k_1+k_2+1)}_{ji},\\ 
\label{PP2}
[P^{2k_1}_{ij},P^{2k_2+2}_{ij}]\simeq 0\quad\text{for}\quad k_1\ge 1,\\ 
\label{PP3}
[P^{2k_1+1}_i,P^{2k_2+1}_{j}]\simeq P^{2(k_1+k_2+1)}_{ij},\\
\label{PP4}
[P^{2k_1+1}_i,P^{2k_2+2}_{ij}]\simeq P^{2(k_1+k_2)+3}_{j},\\ 
\label{PP5}
[P^{2k_1+1}_{i},P^{2k_2+2}_{i'j'}]\simeq 0,\\ 
\label{PP6}
[P^{2k_1}_{ij},P^{2k_2+2}_{i'j'}]\simeq 0\quad\text{for}\quad k_1\ge 1,\\ 
\label{PP7}
[P^{2k_1}_{ij},P^{2k_2+2}_{ij'}]\simeq -P^{2(k_1+k_2+1)}_{jj'}\quad\text{for}\quad k_1\ge 1.
\end{gather}
\end{lemma}

\end{proof}

\section{Miura type transformations} 
\label{miura}

The definition of Miura type transformations was given in the introduction. 
Consider an evolutionary system of order $d\ge 1$ 
\begin{equation}
\label{sys1}
u^i_t=F^i(u^1,\dots,u^m,\,u^1_1,\dots,u^m_1,\dots,u^1_d,\dots,u^m_d),\quad 
i=1,\dots,m,\quad u^i_k=\frac{\pd^k u^i}{\pd x^k}.
\end{equation}
Suppose that for this system we have a ZCR of the form 
\begin{gather}
\label{mn1}
[D_x+M,D_t+N]=D_x(N)-D_t(M)+[M,N]=0,\\ 
\notag
M=M(u^1,\dots,u^m),\quad 
N=(u^1,\dots,u^m,\,u^1_1,\dots,u^m_1,\dots,u^1_{d-1},\dots,u^m_{d-1}) 
\end{gather}
with values in a Lie algebra $\mg$. 
According to Proposition~\ref{wezcr}, 
such ZCR is determined 
by a homomorphism from the WE algebra to $\mg$. 

Consider a homomorphism $\rho$ from $\mg$ to the Lie algebra  
of vector fields on an $m$-dimensional manifold $W$. 
Let $w^1,\dots,w^m$ be local coordinates in $W$ and 
\begin{gather}
\label{rm}
\rho(M)=\sum_{i=1}^m a^i(w^1,\dots,w^m,u^1,\dots,u^m)\frac{\pd}{\pd w^i},\\ 
\label{rn}
\rho(N)=\sum_{i=1}^m b^i(w^1,\dots,w^m,u^1,\dots,u^m,\dots,u^1_{d-1},\dots,u^m_{d-1})\frac{\pd}{\pd w^i}. 
\end{gather}
Then relation~\er{mn1} implies that the system 
\begin{gather}
\label{wx}
\frac{\pd w^i}{\pd x}=a^i(w^1,\dots,w^m,u^1,\dots,u^m),\\ 
\label{wt}
\frac{\pd w^i}{\pd t}=b^i(w^1,\dots,w^m,u^1,\dots,u^m,\dots,u^1_{d-1},\dots,u^m_{d-1}),\\
\notag
u^i_t=F^i(u^1,\dots,u^m,\dots,u^1_d,\dots,u^m_d),\quad i=1,\dots,m, 
\end{gather}
is consistent. 
Suppose that from equations~\er{wx} one can express locally 
\begin{equation}
\label{uc}
u^i=c^i\bigl(w^1,\dots,w^m,\frac{\pd w^1}{\pd x},\dots,\frac{\pd w^m}{\pd x}\bigl),\quad i=1,\dots,m. 
\end{equation}
Substituting this to~\er{wt}, we obtain an evolutionary system 
for $w^1(x,t),\dots,w^m(x,t)$ connected by a Miura type transformation~\er{uc} 
with system~\er{sys1}. 
This is the simplest case of a more general method 
to obtain Miura type transformations from ZCRs~\cite{mtt}. 

Let us apply this construction to system~\er{pt} and its ZCR~\er{M},~\er{N},
where $\la_1,\dots,\la_n$ are again treated as complex parameters satisfying equations~\er{curve} 
and $S=(s^1,\dots,s^n)$ is given by formulas~\er{sp}. 
This ZCR takes values in the Lie subalgebra $\mathfrak{so}_{n,1}\subset\mathfrak{gl}_{n+1}(\Com)$ 
generated by the elements $A_i=E_{i,n+1}+E_{n+1,i}$, $i=1,\dots,n$.

Consider the canonical action of the group $\mathrm{GL}_{n+1}(\Com)$ on 
the projective space $\Com P^n$. It determines a homomorphism $\vf$
from $\mathfrak{gl}_{n+1}(\Com)$ to the Lie algebra of vector fields on $\Com P^n$.
In the standard coordinates $w^1,\dots,w^n$ of an affine chart $\Com^{n}$ 
of $\Com P^n$ this homomorphism reads 
$$
\vf\colon E_{ij}\mapsto w^i\sum_{l=1}^n(\delta_{j,l}-\delta_{j,n+1}w^l)\frac{\pd}{\pd w^l} 
\quad\text{for }\ i,j=1,\dots,n+1,\,\ w^{n+1}=1. 
$$
In particular, 
\begin{equation}
\label{vfA}
\vf(A_i)=\frac{\pd}{\pd w^i}-w^i\sum_{l=1}^nw^l\frac{\pd}{\pd w^l}.
\end{equation}
Clearly, vector fields~\er{vfA} are tangent to the $(n-1)$-dimensional submanifold 
$$
W=\{(w^1,\dots,w^n)\,|\,(w^1)^2+\dots+(w^n)^2=1\}\subset \Com^{n}. 
$$ 
On some neighborhood $W'\subset W$ of the point $(w^1=\dots=w^{n-1}=0,\,w^n=1)\in W$ 
we can take $w^1,\dots,w^{n-1}$ as local coordinates and 
\begin{equation}
\label{wn}
w^n=\sqrt{1-(w^1)^2-\dots-(w^{n-1})^2}.
\end{equation}
Thus we obtain the following homomorphism $\rho$ 
from $\mathfrak{so}_{n,1}$ to the Lie algebra of vector fields on $W'$ 
\begin{gather*}
\rho(A_i)=\vf(A_i)\Bigl|_{W'}=\frac{\pd}{\pd w^i}-w^i\sum_{l=1}^{n-1}w^l\frac{\pd}{\pd w^l},
\quad i=1,\dots,n-1,\\ 
\rho(A_{n})=\vf(A_{n})\Bigl|_{W'}=-\sqrt{1-(w^1)^2-\dots-(w^{n-1})^2}\,\sum_{l=1}^{n-1}w^l\frac{\pd}{\pd w^l}.
\end{gather*}
Applying this to~\er{M}, one gets 
\begin{equation}
\label{rm-new}
\rho(M)=\sum_{i=1}^{n-1}\Bigl(s^i\la_i-w^i\sum_{j=1}^ns^j\la_jw^j\Bigl)\frac{\pd}{\pd w^i},
\end{equation}
where one assumes~\er{wn},~\er{sp}.
Recall that the right-hand side of~\er{wx} is obtained from~\er{rm}. Using~\er{rm-new},~\er{sp},~\er{wn}, 
we get the following form of equations~\er{wx} in our case
\begin{multline}
\label{wx-new}
\frac{\pd w^i}{\pd x}=\\
\frac{1}{1+\sum_{j=1}^{n-1}(u^j)^2}\Biggl(2u^i\la_i-2w^i\sum_{j=1}^{n-1}\la_ju^jw^j-w^i\la_n\sqrt{1-(w^1)^2-\dots-(w^{n-1})^2}
\Bigl(1-\sum_{j=1}^{n-1}(u^j)^2\Bigl)\Biggl),\\
i=1,\dots,n-1.
\end{multline}
Computing $\rho(N)$ by means of~\er{N} and~\er{rm-new}, 
from the coefficients of the vector field~\er{rn} we obtain equations~\er{wt} of the following form 
\begin{multline}
\label{rn-new}
\frac{\pd w^i}{\pd t}=
\sum_{i=1}^{n-1}\Biggl(s^i_{xx}\la_i-w^i\sum_{j=1}^ns^j_{xx}\la_jw^j+\sum_{j=1}^n\la_i\la_jw^j(s^j_xs^i-s^i_xs^j)+\\
\bigl(r_1+\la_1^2 +\frac12(S,RS)+\frac32(S_x,S_x)\bigl)\bigl(s^i\la_i-w^i\sum_{j=1}^ns^j\la_jw^j\bigl)\Biggl)\frac{\pd}{\pd w^i},\\
i=1,\dots,n-1,
\end{multline}
where $S=(s^1,\dots,s^n)$ is given by~\er{sp} and $w^n$ is equal to~\er{wn}.

Denote by $a^i=a^i(w^1,\dots,w^{n-1},u^1,\dots,u^{n-1})$ the right-hand side of~\er{wx-new}. 
It is easily seen that the $(n-1)\times(n-1)$-matrix $||\pd a^i/\pd u^j||$ at the point 
\begin{equation}
\label{wu0}
w^1=\dots=w^{n-1}=u^1=\dots=u^{n-1}=0
\end{equation} 
is equal to $\mathrm{diag}(2\la_1,\dots,2\la_{n-1})$.  
Suppose that  
\begin{equation}
\label{la0}
\la_i\neq 0,\quad i=1,\dots,n-1,
\end{equation} 
then, by the implicit function theorem, 
on a neighborhood of the point~\er{wu0} from equations~\er{wx-new} we can express 
\begin{equation}
\label{uc-new}
u^i=c^i\bigl(w^1,\dots,w^{n-1},\frac{\pd w^1}{\pd x},\dots,\frac{\pd w^{n-1}}{\pd x}\bigl),\quad i=1,\dots,n-1.
\end{equation} 
Substituting~\er{uc-new} to~\er{rn-new}, 
one gets an evolutionary system of the form 
\begin{equation}
\label{wih}
w^i_t=H^i(w^1,\dots,w^{n-1},\dots,w^1_{xxx},\dots,w^{n-1}_{xxx}),\quad i=1,\dots,n-1,
\end{equation}
connected with system~\er{pt} by the Miura type transformation~\er{uc-new}. 

It remains to study the case when~\er{la0} does not hold. 
That is, $\la_k=0$ for some $1\le k\le n-1$. Since in equations~\er{curve} one has $r_i\neq r_j$ for $i\neq j$, 
we obtain that $\la_i\neq 0$ for $i\neq k$. 
Then at any point of the form
\begin{equation}
\label{uwl}
u^k=1,\quad u^l=w^l=0 \quad\forall\,l\neq k
\end{equation}
one has 
$$
\frac{\pd a^i}{\pd u^j}=0\quad\forall\,i\neq j,\quad
\frac{\pd a^i}{\pd u^i}=2\la_i\quad\forall\,i\neq k,\quad
\frac{\pd a^k}{\pd u^k}=2w^k\la_n\sqrt{1-(w^k)^2}.
$$ 
Therefore, at any point of the form~\er{uwl} with~$w^k\neq 0,\pm 1$ 
 the matrix $||\pd a^i/\pd u^j||$ is nonsingular, and on a neighborhood of such point from 
equations~\er{wx-new} we can again get expressions of the form~\er{uc-new} and proceed as described above. 

\section*{Appendix: proof of Lemma~\ref{lemma}}

We prove this by induction on $k_1+k_2$. 
For $k_1+k_2=0$ (that is, $k_1=k_2=0$) the statements follow 
easily  from~\er{rel1} and~\er{rel2}. 
Assume that all the statements are valid for $k_1+k_2\le m$ for some $m\in\zp$.  
We must prove them for $k_1+k_2=m+1$. 

Below $l$ is an arbitrary integer such 
that $1\le l\le n$, $l\neq i$, $l\neq j$. 
We use the notation from Section~\ref{repres}.

\textbf{Proof of~\er{PP0}.} 
First, note that by the induction assumption for $q\le m$ one has 
\begin{equation*}
[ll\ub{i\dots i}_{2q}j]=[ll[P^1_iP^{2q}_{ij}]]\simeq 
[llP^{2q+1}_j]=[l[P^1_lP^{2q+1}_{j}]]\simeq 
[lP^{2q+2}_{lj}]=[P^1_lP^{2q+2}_{lj}]\simeq P^{2q+3}_j.  
\end{equation*}
Similarly, we obtain  
\begin{equation} 
\label{lij}
[ll\ub{i\dots i}_{2q}j]\simeq [\ub{i\dots i}_{2q+2}j]\simeq   
[\ub{l\dots l}_{2q+2}j]\simeq P^{2q+3}_j\quad\forall\,q\le m. 
\end{equation}

Without loss of generality we can assume $k_2\ge 1$ in~\er{PP0}. 
Using\er{lij} and the Jacobi identity, one gets  
\begin{multline}
\label{l1}
[[\ub{i\dots i}_{2k_1}j][\ub{i\dots i}_{2k_2}j]]\simeq 
[[\ub{i\dots i}_{2k_1}j][ll\ub{i\dots i}_{2k_2-2}j]]=
-[[l\ub{i\dots i}_{2k_1}j][li\dots ij]]+[l[\ub{i\dots i}_{2k_1}j][li\dots ij]]\simeq\\
[[ll\ub{i\dots i}_{2k_1}j][i\dots ij]]-2[l[l\ub{i\dots i}_{2k_1}j][i\dots ij]]. 
\end{multline}
To obtain the last line, we used the fact that 
$$
[l[\ub{i\dots i}_{2k_1}j][l\ub{i\dots i}_{2k_2-2}j]]\simeq-[l[l\ub{i\dots i}_{2k_1}j][i\dots ij]],
$$
because $[[\ub{i\dots i}_{2k_1}j][\ub{i\dots i}_{2k_2-2}j]]\simeq 0$ 
by the induction assumption. 

Since, by~\er{lij}, 
$[ll\ub{i\dots i}_{2k_1}j]\simeq [\ub{i\dots i}_{2k_1+2}j]$, from~\eqref{l1} 
one obtains 
\begin{equation*}
[[\ub{i\dots i}_{2k_1}j][\ub{i\dots i}_{2k_2}j]]\simeq 
[[\ub{i\dots i}_{2k_1+2}j][\ub{i\dots i}_{2k_2-2}j]]-2[l[l\ub{i\dots i}_{2k_1}j][\ub{i\dots i}_{2k_2-2}j]].
\end{equation*} 
If $k_2\ge 2$, 
applying the same procedure to $[[\ub{i\dots i}_{2k_1+2}j][\ub{i\dots i}_{2k_2-2}j]]$
yields
\begin{equation*}
[[\ub{i\dots i}_{2k_1}j][\ub{i\dots i}_{2k_2}j]]\simeq 
[[\ub{i\dots i}_{2k_1+4}j][\ub{i\dots i}_{2k_2-4}j]]
-2[l[l\ub{i\dots i}_{2k_1}j][\ub{i\dots i}_{2k_2-2}j]]
-2[l[l\ub{i\dots i}_{2k_1+2}j][\ub{i\dots i}_{2k_2-4}j]].
\end{equation*}
Thus applying this procedure several times to the first summand of the right-hand side, 
we obtain 
\begin{equation}
\label{ijl}
[[\ub{i\dots i}_{2k_1}j][\ub{i\dots i}_{2k_2}j]]\simeq 
[[\ub{i\dots i}_{2(k_1+k_2)}j]j]
-2\sum_{s=1}^{k_2} [l[l\ub{i\dots i}_{2(k_1+s-1)}j][\ub{i\dots i}_{2(k_2-s)}j]].
\end{equation}
By the induction assumption and~\er{lij}, one has 
$$
[[l\ub{i\dots i}_{2(k_1+s-1)}j]i]\simeq 
[[P^1_lP^{2(k_1+s)-1}_j]P^1_i]\simeq [P^{2(k_1+s)}_{lj}P^1_i]\simeq 0.
$$
Therefore, 
\begin{equation}
\label{ijl1}
[l[l\ub{i\dots i}_{2(k_1+s-1)}j][\ub{i\dots i}_{2(k_2-s)}j]]\simeq 
[l\ub{i\dots i}_{2(k_2-s)}[l\ub{i\dots i}_{2(k_1+s-1)}j]j]]=
-[l\ub{i\dots i}_{2(k_2-s)}jl\ub{i\dots i}_{2(k_1+s-1)}j]
\end{equation}
Since, by the induction assumption and~\er{lij}, 
\begin{gather*}
[l\ub{i\dots i}_{2(k_2-s)}j]\simeq [P^1_lP^{2(k_2-s)+1}_j]\simeq 
[P^{2(k_2-s+1)}_{lj}]\simeq-[P^1_jP^{2(k_2-s)+1}_l]\simeq 
-[j\ub{i\dots i}_{2(k_2-s)}l],\\ 
[ll\ub{i\dots i}_{2(k_1+s-1)}j]\simeq [\ub{i\dots i}_{2(k_1+s)}j], 
\end{gather*}
equation~\eqref{ijl1} implies 
\begin{equation*}
[l[l\ub{i\dots i}_{2(k_1+s-1)}j][\ub{i\dots i}_{2(k_2-s)}j]]\simeq
-\bigl[l\ub{i\dots i}_{2(k_2-s)}jl\ub{i\dots i}_{2(k_1+s-1)}j\bigl]\simeq
\bigl[j\ub{i\dots i}_{2(k_2-s)}ll\ub{i\dots i}_{2(k_1+s-1)}j\bigl]\simeq 
\bigl[j\ub{i\dots i}_{2(k_1+k_2)}j\bigl]
\end{equation*}
Combining this with~\eqref{ijl} yields 
\begin{equation}
\label{ijl2}
\bigl[[\ub{i\dots i}_{2k_1}j][\ub{i\dots i}_{2k_2}j]\bigl]\simeq 
-\bigl[j\ub{i\dots i}_{2(k_1+k_2)}j\bigl]
-2k_2\bigl[j\ub{i\dots i}_{2(k_1+k_2)}j\bigl]. 
\end{equation}
For $k_1=0$ this equation implies $[j\ub{i\dots i}_{2(k_1+k_2)}j\bigl]\simeq 0$. 
Combing this with~\er{ijl2}, we obtain~\er{PP0}.

\textbf{Proof of~\er{PP1}.} 
By the induction assumption, \er{PP0}, and~\er{rel1}, 
\begin{gather*}
[\ub{i\dots i}_{2m+1}j]\simeq -[\ub{j\dots j}_{2m+1}i],\quad 
[i\ub{j\dots j}_{2m}i]\simeq 0,\\ 
[[iij]j]\simeq 0,\quad [iij]\simeq [llj],\quad [jji]\simeq [lli].   
\end{gather*}
Using this and~\er{rel2}, one gets 
\begin{multline*}
P^{2(k_1+k_2+1)}_{ij}=
[\ub{i\dots i}_{2m+3}j]=[ii\ub{i\dots i}_{2m+1}j]\simeq -[ii\ub{j\dots j}_{2m+1}i]\simeq\\ 
-[[iij]\ub{j\dots j}_{2m}i]\simeq -[\ub{j\dots j}_{2m}[iij]i]=[\ub{j\dots j}_{2m}iiij]\simeq  
[\ub{j\dots j}_{2m}illj]=[\ub{j\dots j}_{2m}[il]lj]=\\
[\ub{j\dots j}_{2m}[[il]l]j]=
-[\ub{j\dots j}_{2m}jlli]\simeq -[\ub{j\dots j}_{2m}jjji]=-P^{2(k_1+k_2+1)}_{ji}.
\end{multline*}

\textbf{Proof of~\er{PP2}.} 
By the Jacobi identity, 
\begin{equation}
\label{PP2p}
[P^{2k_1}_{ij},P^{2k_2+2}_{ij}]=[[\ub{i\dots i}_{2k_1}j][\ub{i\dots i}_{2k_2+1}j]]=
-[[\ub{i\dots i}_{2k_1}j][\ub{i\dots i}_{2k_2}j]]+
[i[\ub{i\dots i}_{2k_1-1}j][\ub{i\dots i}_{2k_2}j]].  
\end{equation}
By the induction assumption and~\er{lij}, 
$$
[[\ub{i\dots i}_{2k_1-1}j][\ub{i\dots i}_{2k_2}j]]\simeq [P^{2k_1}_{ij}P^{2k_2+1}_{j}]\simeq 
P^{2(k_1+k_2)+1}_i.
$$
Substituting this to~\er{PP2p} and using~\er{PP0}, we obtain  
\begin{equation*}
[P^{2k_1}_{ij},P^{2k_2+2}_{ij}]\simeq
-[[\ub{i\dots i}_{2k_1}j][\ub{i\dots i}_{2k_2}j]]+[i,P^{2(k_1+k_2)+1}_i]\simeq 0. 
\end{equation*}

\textbf{Proof of~\er{PP3}.} 
By~\er{lij} and the induction assumption of~\er{PP5}, 
for any $q_1,\,q_2\in\mathbb{Z}_+$ 
such that $q_1+q_2\le 2(k_1+k_2)-1$ and $q_1+q_2$ is odd  one has 
$$
[[\ub{l\dots l}_{q_1}i][\ub{l\dots l}_{q_2}j]]\simeq 0. 
$$
Using this and the Jacobi identity, we get  
\begin{multline*}
[[\ub{l\dots l}_{2k_1}i][\ub{l\dots l}_{2k_2}j]]\simeq 
-[[\ub{l\dots l}_{2k_1+1}i][\ub{l\dots l}_{2k_2-1}j]]\simeq \\ 
[[\ub{l\dots l}_{2k_1+2}i][\ub{l\dots l}_{2k_2-2}j]]\simeq\dots 
\simeq [[\ub{l\dots l}_{2(k_1+k_2)}i]j]=-[j\ub{l\dots l}_{2(k_1+k_2)}i]  
\end{multline*}
Combining this with~\eqref{lij} and~\er{PP1} yields 
\begin{multline*}
[P^{2k_1+1}_i,P^{2k_2+1}_{j}]\simeq [[\ub{l\dots l}_{2k_1}i][\ub{l\dots l}_{2k_2}j]]\simeq \\
-[j\ub{l\dots l}_{2(k_1+k_2)}i]\simeq -[j\ub{j\dots j}_{2(k_1+k_2)}i]=
-P^{2(k_1+k_2+1)}_{ji}\simeq P^{2(k_1+k_2+1)}_{ij}. 
\end{multline*}

\textbf{Proof of~\er{PP4}.} 
Consider first the case $k_1=0$, then 
\begin{equation*}
[P^{2k_1+1}_i,P^{2k_2+2}_{ij}]=[i\ub{i\dots i}_{2k_2+1}j]. 
\end{equation*}
Set $c=1$ if $j\neq 1$ and $c=2$ if $j=1$. 
If $i=c$ then $[i\ub{i\dots i}_{2k_2+1}j]=P_j^{2(k_1+k_2)+3}$ 
and the proof is complete. 

Suppose that $i\neq c$, then using~\er{PP1} one gets 
\begin{equation*}
[P^{2k_1+1}_i,P^{2k_2+2}_{ij}]=[i\ub{i\dots i}_{2k_2+1}j]
\simeq -[i\ub{j\dots j}_{2k_2+1}i].  
\end{equation*}
Since, by ~\er{PP0}, 
$[i\ub{j\dots j}_{2k_2}i]\simeq 0$, using~\er{lij} and $[[ij]c]=0$ 
we obtain 
\begin{multline*}
[P^{2k_1+1}_i,P^{2k_2+2}_{ij}]\simeq -[i\ub{j\dots j}_{2k_2+1}i]\simeq 
-[[ij]\ub{j\dots j}_{2k_2}i]\simeq -[[ij]\ub{c\dots c}_{2k_2}i]=\\
-[\ub{c\dots c}_{2k_2}[ij]i]=[\ub{c\dots c}_{2k_2}[iij]]\simeq 
[\ub{c\dots c}_{2k_2}[ccj]]=P_j^{2(k_1+k_2)+3}.  
\end{multline*}

Now consider the case $k_1\ge 1$. 
By~\er{lij}, for any $l\neq i$, $l\neq j$ we have
\begin{equation*}
[P^{2k_1+1}_i,P^{2k_2+2}_{ij}]=-[P^{2k_2+2}_{ij},P^{2k_1+1}_i]\simeq
-[[\ub{i\dots i}_{2k_2+1}j][\ub{l\dots l}_{2k_1}i].  
\end{equation*}
Since due to the induction assumption of~\er{PP5} one has $[[\ub{i\dots i}_{2k_2+1}j]l]\simeq 0$, 
using~\er{lij} we get  
\begin{equation*}
[P^{2k_1+1}_i,P^{2k_2+2}_{ij}]\simeq
-[[\ub{i\dots i}_{2k_2+1}j][\ub{l\dots l}_{2k_1}i]\simeq 
-[\ub{l\dots l}_{2k_1}[\ub{i\dots i}_{2k_2+1}j]i]=
[\ub{l\dots l}_{2k_1}\ub{i\dots i}_{2k_2+2}j]\simeq 
[\ub{l\dots l}_{2(m+2)}j].  
\end{equation*}
In the consideration of the case $k_1=0$ we showed that for any $i\neq j$ 
$[\ub{i\dots i}_{2(m+2)}j]\simeq P_j^{2m+5}$. (Recall that $m=k_1+k_2-1$).  
Therefore, 
$$
[\ub{l\dots l}_{2(m+2)}j]\simeq P_j^{2(k_1+k_2)+3}.
$$ 

\textbf{Proof of~\er{PP5}.} 
Consider first the case $k_1\ge 1$. 
Since $n\ge 4$, there is $l\in\{1,\dots,n\}$ such that 
$l\neq i$, $l\neq i'$, $l\neq j'$. 
Then, by the induction assumption of~\er{PP5}, 
$$
[[\ub{i'\dots i'}_{2k_2+1}j']l]\simeq 0,\quad 
[[\ub{i'\dots i'}_{2k_2+1}j']i]\simeq 0. 
$$ 
Using this and~\er{lij}, one gets 
\begin{equation}
\label{PlP}
[P^{2k_1+1}_{i},P^{2k_2+2}_{i'j'}]=-[P^{2k_2+2}_{i'j'},P^{2k_1+1}_{i}]\simeq 
[[\ub{i'\dots i'}_{2k_2+1}j']\ub{l\dots l}_{2k_1}i]\simeq 
[l\dots l[\ub{i'\dots i'}_{2k_2+1}j']i]\simeq 0. 
\end{equation}
If we set $k_2=0$ then~\er{PlP} implies that for any distinct integers 
$c_1,c_2,c_3,c_4\in\{1,\dots,n\}$ 
\begin{equation}
\label{ccc'}
[[c_1c_2]\ub{c_4\dots c_4}_{2m+2}c_3]\simeq 0. 
\end{equation}
Since $[\ub{c_4\dots c_4}_{2m+2}c_3]\simeq [\ub{c_2\dots c_2}_{2m+2}c_3]$ due to~\er{lij}, 
from~\er{ccc'} we obtain 
\begin{equation}
\label{ccc}
[[c_1c_2]\ub{c_2\dots c_2}_{2m+2}c_3]\simeq 0. 
\end{equation}
By the Jacobi identity, \er{ccc}, and~\er{lij},  
\begin{equation}
\label{ccc1}
[c_1\ub{c_2\dots c_2}_{2m+3}c_3]\simeq 
[c_2c_1\ub{c_2\dots c_2}_{2m+2}c_3]\simeq 
[c_2\ub{c_1\dots c_1}_{2m+3}c_3].
\end{equation}
Also, property~\er{PP1} implies 
\begin{equation}
\label{ccc2}
[c_1\ub{c_2\dots c_2}_{2m+3}c_3]\simeq -[c_1\ub{c_3\dots c_3}_{2m+3}c_2].
\end{equation} 
It remains to study the case $k_1=0$. 
Using~\er{ccc1} and~\er{ccc2}, we get 
\begin{multline*}
[P^{2k_1+1}_{i},P^{2k_2+2}_{i'j'}]=
[i\ub{i'\dots i'}_{2k_2+1}j']\simeq 
[i'\ub{i\dots i}_{2k_2+1}j']\simeq
-[i'j'\dots j'i]\simeq 
-[j'i'\dots i'i]\simeq\\
[j'i\dots ii']\simeq
[ij'\dots j'i']\simeq 
-[ii'\dots i'j']=-[P^{2k_1+1}_{i},P^{2k_2+2}_{i'j'}]. 
\end{multline*}
Therefore, $[P^{2k_1+1}_{i},P^{2k_2+2}_{i'j'}]\simeq 0$. 

\textbf{Proof of~\er{PP6}.} By~\er{PP5}, 
$$
[i,P^{2k_2+2}_{i'j'}]\simeq 0,\quad [j,P^{2k_2+2}_{i'j'}]\simeq 0.  
$$ 
Obviously, this implies~\er{PP6}. 

\textbf{Proof of~\er{PP7}.} 
By property~\er{PP5}, $[[\ub{i\dots i}_{2k_1-1}j]j']\simeq 0$. 
Using this, equation~\er{lij}, and~\er{PP1}, one obtains
\begin{multline*} 
[P^{2k_1}_{ij},P^{2k_2+2}_{ij'}]\simeq -[P^{2k_1}_{ij},P^{2k_2+2}_{j'i}]=
-[[\ub{i\dots i}_{2k_1-1}j]\ub{j'\dots j'}_{2k_2+1}i]\simeq
-[\ub{j'\dots j'}_{2k_2+1}[\ub{i\dots i}_{2k_1-1}j]i]=\\
[\ub{j'\dots j'}_{2k_2+1}\ub{i\dots i}_{2k_1}j]\simeq
[\ub{j'\dots j'}_{2k_2+1}\ub{j'\dots j'}_{2k_1}j]=P^{2(k_1+k_2+1)}_{j'j}\simeq 
-P^{2(k_1+k_2+1)}_{jj'}.
\end{multline*}

\section*{Acknowledgements}
The authors thank T.~Skrypnyk and V.~V.~Sokolov for helpful discussions. 
Work of S.I. is supported by the NWO VENI grant 639.031.515.  
JvdL is partially supported by the European Union through the FP6 Marie
Curie Grant (ENIGMA) and the European Science Foundation (MISGAM). 
S.I. thanks the Max Planck Institute for Mathematics (Bonn, Germany) 
for its hospitality and excellent working conditions during  02.2006-01.2007, 
when a part of this research was done.

\end{document}